\documentclass[11pt,letter]{article}
\usepackage[margin=1in]{geometry}
\usepackage{libertine}
\usepackage{xfrac}
\usepackage{bm}
\usepackage{datetime}
\usepackage{hhline}
\usepackage{multirow}
\usepackage{stmaryrd}
\usepackage{booktabs}
\usepackage{enumitem}
\usepackage{wrapfig}
\usepackage{subcaption}
\usepackage[T1]{fontenc}
\usepackage[utf8]{inputenc}
\usepackage{pgfplots}
\usepackage{amsmath, amssymb}
\usepackage{amsthm}
\usepackage{color}
\usepackage{cases}
\usepackage{xparse}
\usepackage{xargs}
\usepackage{appendix}
\usepackage[boxed,commentsnumbered,noend,linesnumbered]{algorithm2e}
\usepackage{algpseudocode}
\usepackage{verbatim, xspace}
\usepackage{tikz}
\usepackage{mathtools}
\usepackage{thm-restate}
\usepackage[framemethod=TikZ]{mdframed}

\SetCommentSty{FontInPseudocodeComments}
\usepackage{hyperref}
\usepackage{pgfplots}

\usepackage{xcolor}
\usepackage[noabbrev,capitalize]{cleveref}

\newcommand{\citet}[1]{\cite{#1}}

\usepackage[colorinlistoftodos,prependcaption,textsize=tiny]{todonotes}
\newcommandx{\unsure}[2][1=]{\todo[linecolor=green,backgroundcolor=green!25,bordercolor=green,#1]{\normalsize #2}}
\newcommandx{\improvement}[2][1=]{\todo[inline,linecolor=blue,backgroundcolor=blue!05,bordercolor=blue,#1]{\normalsize #2}}
\newcommandx{\info}[2][1=]{\todo[linecolor=yellow,backgroundcolor=yellow!25,bordercolor=yellow,#1]{#2}}
\newcommandx{\floatmodel}[2][1=]{\todo[inline,linecolor=red,backgroundcolor=yellow!25,bordercolor=yellow,#1]{#2}}
\newcommandx{\thiswillnotshow}[2][1=]{\todo[disable,#1]{#2}}
\newcommandx{\karol}[2][1=]{\todo[inline,linecolor=blue,backgroundcolor=blue!25,bordercolor=blue,caption={\normalsize \textbf{Karol}},#1]{\normalsize #2}}
\newcommandx{\anita}[2][1=]{\todo[inline,linecolor=gray,backgroundcolor=red!25,bordercolor=red,caption={\normalsize
\textbf{Anita}},#1]{\normalsize #2}}
\newcommandx{\manita}[2][1=]{\todo[linecolor=gray,backgroundcolor=red!25,bordercolor=red,caption={\normalsize
\textbf{A}},#1]{\normalsize #2}}
\newcommandx{\evangelos}[2][1=]{\todo[inline,linecolor=gray,backgroundcolor=green!25,bordercolor=green,caption={\normalsize
\textbf{Evangelos}},#1]{\normalsize #2}}
\newtheorem{theorem}{Theorem}
\newtheorem{definition}[theorem]{Definition}
\newtheorem{lemma}[theorem]{Lemma}
\newtheorem{corollary}[theorem]{Corollary}
\newtheorem{claim}[theorem]{Claim}

\newtheorem{observation}[theorem]{Observation}

\newtheorem{remark}[theorem]{Remark}

\newtheorem*{maingoal*}{Main Question}

\numberwithin{theorem}{section}
\numberwithin{lemma}{section}
\numberwithin{claim}{section}
\numberwithin{corollary}{section}
\numberwithin{definition}{section}
\numberwithin{observation}{section}
\numberwithin{proposition}{section}

\newcommand{\dc}{{\downarrow}}

\newcommand{\eps}{\varepsilon}

\newcommand{\Oh}{\mathcal{O}}
\newcommand{\Os}{\Oh^{\star}}
\newcommand{\Otilde}{\widetilde{\Oh}}
\newcommand{\Ot}{\Otilde}

\newcommand{\nat}{\mathbb{N}}

\newcommand{\Aa}{\normalfont{{A}}}
\newcommand{\Bb}{\normalfont{{B}}}
\newcommand{\Ff}{\normalfont{{F}}}
\newcommand{\Rr}{\normalfont{{R}}}
\newcommand{\Ll}{\normalfont{{L}}}

\newcommand{\Ss}{\normalfont{{S}}}
\newcommand{\Cc}{\normalfont{{C}}}

\newcommand{\poly}{\mathrm{poly}}

\newcommand{\Prob}[1]{\mathbb{P} \left[ #1 \right]}
\newcommand{\Ex}[1]{\mathbb{E}\left[ #1 \right]}

\newcommand{\iv}[1]{\llbracket #1 \rrbracket}

\newcommand{\veca}{\normalfont{\textbf{a}}}
\newcommand{\vecb}{\normalfont{\textbf{b}}}
\newcommand{\vecc}{\normalfont{\textbf{c}}}
\newcommand{\vecx}{\normalfont{\textbf{x}}}
\newcommand{\vecy}{\normalfont{\textbf{y}}}
\newcommand{\vecu}{\normalfont{\textbf{u}}}
\newcommand{\vecs}{\normalfont{\textbf{s}}}
\newcommand{\vecr}{\normalfont{\textbf{r}}}

\renewcommand{\leq}{\leqslant}
\renewcommand{\geq}{\geqslant}
\renewcommand{\le}{\leqslant}
\renewcommand{\ge}{\geqslant}
\renewcommand{\subset}{\subseteq}
\renewcommand{\epsilon}{\varepsilon}

\title{Faster algorithms for $k$-Orthogonal Vectors in low dimension}
\date{}

\author{
  Anita D\"urr\footnote{ETH Z\"urich, Z\"urich, Switzerland, \textsf{anita.duerr@inf.ethz.ch}. {Part of this work was done while affiliated to Saarland University and Max Planck Institute for Informatics, Saarbrücken, Germany, where this work was part of the project TIPEA that has received funding from the European Research Council (ERC) under the European Unions Horizon 2020 research and innovation programme (grant agreement No. 850979).}}
    \and
    Evangelos Kipouridis\footnote{Max Planck Institute for Informatics,
        Saarbr\"ucken, Germany, \textsf{kipouridis@mpi-inf.mpg.de}.}
    \and
    Michael Lampis\footnote{Université Paris Dauphine, Paris, France, \textsf{michail.lampis@lamsade.dauphine.fr}}
    \and
    Karol W\k{e}grzycki\footnote{Max Planck Institute for Informatics,
        Saarbr\"ucken, Germany, \textsf{kwegrzyc@mpi-inf.mpg.de}. Supported by the Deutsche Forschungsgemeinschaft (DFG, German Research Foundation) grant number 559177164. }
}
\begin{document}

\maketitle

\thispagestyle{empty}
\begin{abstract}
    In the Orthogonal Vectors problem (OV), we are given two families $\Aa,
    \Bb$ of subsets of $\{1,\ldots,d\}$, each of size $n$, and the task is to decide whether
    there exists a pair $\veca \in \Aa$ and $\vecb \in \Bb$ such
    that $\veca \cap \vecb =
    \emptyset$. Straightforward algorithms for this problem run in
    $\mathcal{O}(n^2 \cdot d)$ or $\Oh(2^d \cdot n)$ time, and assuming SETH, there is no $2^{o(d)}\cdot n^{2-\varepsilon}$ time
    algorithm that solves this problem for any constant $\varepsilon > 0$.

    Williams (FOCS 2024) presented a $\tilde{\mathcal{O}}(1.35^d \cdot n)$-time algorithm for the
    problem, based on the succinct equality-rank decomposition of the disjointness
    matrix. In this paper, we present a combinatorial algorithm that runs in
    randomized time $\tilde{\mathcal{O}}(1.25^d \cdot n)$. This can be improved to $\Oh(1.16^d \cdot n)$ using computer-aided evaluations.

    We also consider a more general $k$-Orthogonal Vectors problem, where
    given $k$ families $\Aa_1,\ldots,\Aa_k$ of subsets of
    $\{1,\ldots,d\}$, each of size $n$, the task is to find
    elements $\veca_i \in \Aa_i$ for every $i \in \{1,\ldots,k\}$ such
    that $\veca_1 \cap \veca_2
    \cap \ldots \cap \veca_k = \emptyset$. We show that for every fixed $k \ge 2$, there
    exists $\varepsilon_k > 0$ such that the $k$-OV problem can be solved in time
    $\mathcal{O}(2^{(1 - \varepsilon_k)\cdot d}\cdot n)$. We also show that, asymptotically, this is the
    best we can hope for: 
    for any $\eps > 0$ there exists a $k \ge 2$ such that $2^{(1 -
    \varepsilon)\cdot d} \cdot n^{\mathcal{O}(1)}$ time algorithm for $k$-Orthogonal
    Vectors would contradict the Set Cover Conjecture.
    
\end{abstract}

% \begin{picture}(0,0)
% \put(462,-170)
% {\hbox{\includegraphics[width=40px]{img/logo-erc.jpg}}}
% \put(452,-230)
% {\hbox{\includegraphics[width=60px]{img/logo-eu.pdf}}}
% \end{picture}

\clearpage
\setcounter{page}{1}

\section{Introduction}

In the Orthogonal Vectors problem, the task is to find a disjoint pair of vectors in a given collection of vectors.
We view vectors in $\{0, 1\}^d$ as subsets of $\{1, 2, \dots, d\}$ for some dimension $d$.

\begin{definition}[Orthogonal Vectors ($2$-OV)]
    Given two families $\Aa, \Bb$ of subsets of $\{1,\ldots,d\}$ with $|\Aa| = |\Bb|
    = n$, the Orthogonal Vectors problem asks whether there exist $\veca \in
    \Aa$ and $\vecb \in \Bb$ such that $\veca \cap \vecb = \emptyset$.
\end{definition}

The naive algorithm for the Orthogonal Vectors problem runs in $\Oh(n^2\cdot d)$
time.  This quadratic-time algorithm has been slightly improved to
$n^{2-1/\Oh(\log(d/\log
n))}$\cite{chan-williams-16,batch-ov}, but no
$\Oh(n^{2-\eps})$-time algorithm is currently known for any $\eps > 0$. In
fact, the conjecture that there is no such $\eps > 0$ for which Orthogonal
Vectors can be solved in $\Oh(n^{2-\eps})$ 
% \manita{I added a $\Oh$ here but I'm not sure its correct}
time for $d = \omega(\log n)$ is one of
the central hypotheses in the field of fine-grained
complexity\cite{ipec-survey}.  The Orthogonal Vectors problem and its
connections to other problems have been thoroughly investigated~\cite{ChenW19,
GaoIKW19}, and a $2^{o(d)} \cdot n^{2-\eps}$-time algorithm for any $\eps > 0$
would contradict the Strong Exponential Time Hypothesis~\cite{ov-seth-hard}.

In this paper we focus on the low-dimensional regime when  $d = c \log n$ for some small constant $c > 0$. In that regime, the SETH lower bound does not preclude the possibility of subquadratic algorithms.
In fact, folklore algorithms with running times $\Oh(d \cdot 2^d \cdot n)$ or even $\Ot(n + 2^d)$ are known~\cite{bjorklund-ov}\footnote{$\Ot(\cdot)$ hides
polylogarithmic factors, while $\Os(\cdot)$ hides polynomial factors in the
input. Since $d = c \log n$ for some constant $c >0$ throughout this paper, $\Ot$ also hides polynomial factors in $d$.}.

Inspired by connections to non-uniform circuit lower bounds,
Williams~\cite{williams24} investigated whether the dependence on $d$ can be
improved, and designed an $\Ot(1.35^d\cdot n)$-time algorithm for Orthogonal Vectors.  Williams~\cite{williams24} exploited
the structure of the problem by presenting a systematic approach based on
constant-sized decompositions of the disjointness matrix.

Our first contribution is a new combinatorial algorithm for the problem:

\begin{restatable}{thm}{twoov}
    \label{thm:2ov}
    The Orthogonal Vectors problem can be solved in $\Ot(1.25^d\cdot n)$ time by a
    randomized, one-sided error algorithm that succeeds with probability $1 -
    2^{-d^{\Omega(1)}}$. Using computer-aided evaluation, the running time can be improved to $\Oh(1.16^d \cdot n)$.
\end{restatable}

This result was independently discovered by Alman~and~Li~\cite{alman25}. It is worth noting
that the technique of Williams~\cite{williams24} also allows one to
deterministically count the number of solutions to the Orthogonal Vectors
problem in $\Ot(1.38^d \cdot n)$ time, whereas ours cannot count and is
inherently randomized. Our approach is based on an application of the
representation method~\cite{fomin-jacm2016}.  In a nutshell, we observe that if
$\veca \cap \vecb = \emptyset$, then there exists a certificate set $\vecc$ such
that: (i) $\veca \subseteq \vecc$, and (ii) $\vecb \cap \vecc = \emptyset$.  Our
algorithm samples an appropriate number of such sets and constructs a data
structure that certifies conditions (i) and (ii). 

We remark that, using computer-aided evaluation, the running time of our
algorithm can be improved to $\Oh(1.16^d \cdot n)$. This constant is very close
to the barrier achievable with current techniques when the size of $\veca$ and $\vecb$ is $d/4$: in that setting, a similar approach was used to design a $\Oh(1.14^d \cdot n)$ time algorithm for Orthogonal Vectors. Any improvement on that running time would improve known algorithms for Subset Sum~\cite{stoc21}.
% \karol{Anita wants to rewirte this}
% \manita{I changed that sentence. Please check it (I cannot run python code so maybe 1.13 is wrong?)}

Next, we consider a more general $k$-Orthogonal Vectors problem:

\begin{definition}[$k$-Orthogonal Vectors ($k$-OV)]
Given families $\Aa_1,\ldots,\Aa_k$ of subsets of $\{1,\ldots,d\}$, each of size
$n$, the $k$-Orthogonal Vectors problem asks whether there exist $\veca_1 \in
\Aa_1, \ldots, \veca_k \in \Aa_k$ such that $\veca_1 \cap \veca_2 \cap \ldots
\cap \veca_k = \emptyset$.
\end{definition}

Naive algorithms for the $k$-Orthogonal Vectors problem run in $\Oh(d\cdot k\cdot n^k)$
time, and any $2^{o(d)}\cdot n^{k-\eps}$ time algorithm for $\eps > 0$ would
contradict SETH. As in the 2-OV case, for every fixed $k \ge 2$, a folklore
algorithm solves the $k$-Orthogonal Vectors problem in $\Oh(d\cdot 2^d\cdot n)$ time. We
show that the exponential dependence on $d$ can be improved:

\begin{restatable}{thm}{kov}
\label{thm:kov}
    For every $k \ge 2$, there exists $\eps_k > 0$ such that the $k$-Orthogonal
    Vectors problem can be solved deterministically in time $\Oh(2^{(1-\eps_k)d} \cdot n)$.
\end{restatable}

To prove the theorem, we observe that the algorithm of Bj\"orklund et al.~\cite{bjorklund-ov} can
be generalized to give a $\Oh\big(d \cdot (|\dc \Aa_1| + \ldots + |\dc
\Aa_k|)\big)$-time algorithm for the $k$-Orthogonal Vectors problem, where $\dc
\Ff \coloneq \{\vecx \subseteq \vecs \mid \vecs \in \Ff\}$ denotes the
down-closure of the family $\Ff$.  First, guess the cardinalities of the solution $1 \le
\alpha_1 d \le \ldots\le\alpha_k d \le d$ and assume that the family $\Aa_i$ contains only sets of
cardinality $\alpha_i d$.  If $\alpha_k < (1-\eps)$ for some $\eps > 0$, then $|\dc \Aa_i| \leq 2^{(1-\eps)d}
|\Aa_i|$ for each $i \in \{1,\ldots,k\}$ and thus the
above algorithm runs in time $\Oh(2^{(1-\eps)d} \cdot n)$.
Therefore, consider the case where $\alpha_k > (1 - \eps)$. 
Then $|\Aa_k| \le \binom{d}{\eps d}$. Hence, we can afford to guess a set
$\veca_k \in \Aa_k$ and recurse on the $(k-1)$-Orthogonal Vectors problem with universe $\veca_k$. We
continue this process until we reach the case $k = 2$, at which point we can use
a deterministic algorithm for $k=2$.

Finally, we observe that \cref{thm:kov} is asymptotically optimal: having a
$2^{(1-\eps)d} \cdot n^{\Oh(1)}$-time algorithm that solves $k$-Orthogonal Vectors for every $k \ge 2$
would drastically improve the currently best algorithms for Set Cover.

\begin{restatable}{thm}{scclowerbound}\label{thm:lowerBound}
    For every $c > 0$ and $\eps > 0$, there exists $k \in \nat$ such that the
    $k$-Orthogonal Vectors problem is not solvable in $\Oh(2^{(1-\eps)d} \cdot n^c)$ time,
    assuming the Set Cover Conjecture.
\end{restatable}

Recently, it has been shown that the \emph{Set Cover Conjecture} and the \emph{Asymptotic
Rank Conjecture} cannot both be true~\cite{bjorklund24,pratt24}.  
We show that these techniques can also be used to conditionally improve algorithms for $k$-Orthogonal Vectors: 

\begin{restatable}{thm}{arcalg}\label{thm:arcalg}
    Assuming the Asymptotic Rank Conjecture, there exist $c > 0$ and $\eps > 0$, such that
    any $k$-Orthogonal Vectors instance can be solved in $\Oh(2^{(1-\eps)\cdot d} \cdot n^c)$ time for every $k \in \nat$.
\end{restatable}

As these results are relatively straightforward corollaries of the known connections, we include them in~\cref{app:SetCover}.
% However, such
% improvements may be beyond the reach of current techniques.  At present,
% assuming the Asymptotic Rank Conjecture, it is possible to obtain an $\Os(2^{(1 -
% \eps)\cdot n})$-time algorithm for Set Cover on a universe of size $n$ for some $\eps > 0$, when all sets
% have size at most $n/4$~\cite{bjorklund24,pratt24,radu25}.
% On the other hand, the reduction behind~\cref{thm:lowerBound} implies that a
% faster algorithm for $k$-Orthogonal Vectors would lead to a faster algorithm for
% Set Cover in the general case.

% \manita{Any reason this was a subparagraph and not paragraph?}
\paragraph*{Related Work} As mentioned before, techniques used in this paper
are inspired by~\cite{stoc21} which in turn is based on the more general
framework of Fomin et al.~\cite{fomin-jacm2016}. These techniques have also
been used to give a faster algorithm for Subset Balancing Problems~\cite{randolph25}.
Similar techniques have also been used by Chukhin et al.~\cite{DBLP:journals/eccc/ChukhinKMS25} in the context of monotone
circuits and matrix rigidity.

% \manita{I moved this paragraph here, felt a bit lost in the preliminaries section}
\paragraph{Organization} After introducing notations and discussing a folklore algorithm in \cref{sec:preliminaries}, we give our algorithm for the Orthogonal Vectors problem in \Cref{sec:2OV}. In \Cref{sec:kOV} we give our algorithm for the $k$-Orthogonal Vectors problem. Our lower bound for $k$-Orthogonal Vectors is provided in \Cref{app:SetCover}, as the proof is straightforward.

\section{Preliminaries}\label{sec:preliminaries}

% \anita{Decide whether vectors should be capitalized and fix consistency}
% Karol: done

% \manita{TODO: check consistency between OV and Orthogonal Vectors}

We use the shorthand notation $[n] \coloneq \{1, 2, \dots, n\}$ for any $n \in \mathbb N$. 
For a universe set $\vecu$ and a set family $\Ff \subseteq 2^{\vecu}$, we denote
the down-closure of $\Ff$ as $\dc \Ff \coloneq \{ \vecx \subseteq \vecs \mid
\vecs \in \Ff\}$. We use $\vecs
\subseteq_R \Ff$ to denote that $\vecs$ is a subset of $\Ff$ chosen uniformly at random among all such subsets.
For a non-negative integer $i \le |\vecu|$, the set of subsets of $\vecu$ of size $i$ is denoted $\binom{\vecu}{i}$.
% \manita{@Karol: I find it confusing that universe is small u now}
We use $\vecx
\sqcup \vecy$ to denote a disjoint partition. 
We use the Iverson bracket notation for the
characteristic function; i.e.\ for any logical expression $b$, the value of
$\iv{b}$ is $1$ if $b$ is true and $0$ otherwise.

\paragraph{Binomial Coefficients}
The binary entropy function is defined as $h(\alpha) \coloneq -\alpha
\log_2(\alpha)-(1-\alpha)\log_2(1-\alpha)$ for every $\alpha \in (0,1)$, and
$h(\alpha) = 0$ for $\alpha \in \{0,1\}$. We use it to approximate binomial
coefficients with the following inequalities~\cite{robbins1955remark}:
\begin{equation}\label{eq:binom_apx}
    2^{h(\alpha) n} \ge \binom{n}{\alpha n} \ge \Omega\left(2^{h(\alpha) n} \cdot n^{-1/2}\right),
\end{equation}
for every $\alpha \in (0,1)$. The binary entropy function generalizes to the multinomial form. For $\alpha_1,\ldots,\alpha_k \in (0,1)$ with $\sum_{i=1}^k \alpha_i = 1$ it is defined as:
\begin{displaymath}
	h(\alpha_1,\ldots,\alpha_k) \coloneqq -\sum_{i=1}^k \alpha_i \log_2 \alpha_i.
\end{displaymath}
% \manita{can we say this is log base 2?}
Note that for $\alpha \in (0,1)$, we use the shorthand notation
$h(\alpha) \coloneqq h(\alpha,1-\alpha)$.
The multinomial coefficient can be approximated with the function $h$ as follows:
	
\begin{lemma}[\cite{csiszar2004information}, Lemma 2.2]\label{lem:multvsentropy}
	\[
		 2^{h(\alpha_1,\ldots,\alpha_k)n} \cdot n^{-\Oh(k)}\leq \binom{n}{\alpha_1 n,\ldots,\alpha_k n} \leq 2^{h(\alpha_1,\ldots\alpha_k)n}.
	\]
\end{lemma}

\paragraph{Orthogonal Vectors}
Throughout the paper, we denote by $n$ the number of vectors given in each family of an OV instance and by $d$ the dimension of the vectors. Vectors in $\{0, 1\}^d$ are often interpreted as subsets of $[d]$, and we use these terms interchangeably. We always assume that $d = c \log n$ for some constant $c >0$. 
For completeness, we include a simple deterministic algorithm for 2-OV. Notice
that~\cite{williams24} offers a faster deterministic algorithm. However, the running time proven below is enough for the base case of the recursive algorithm in \cref{thm:kov}.

\begin{lemma} \label{thm:det-ov}
    The Orthogonal Vectors problem can be solved deterministically in $\Ot(2^{d/2} \cdot n)$ time.
\end{lemma}
\begin{proof}
    Partition the universe $[d]$ into the low order bits $\vecu_L = \{1, \dots,
    \lfloor d/2 \rfloor\}$ and high order bits $\vecu_H = \{\lfloor d/2 \rfloor + 1, \dots, d\}$.
    For any subset $\vecx \subset [d]$, we write $\vecx = \vecx_L \uplus
    \vecx_H$ for $\vecx_L = \vecx \cap \vecu_L$ and $\vecx_H = \vecx \cap \vecu_H$.
    Consider the following sets of vectors:
    \begin{align*}
    \Ss_A \coloneqq &\left\{ \veca_L \uplus \vecx_H \subseteq [d] \ \mid \ 
    \veca_L \uplus \veca_H \in \Aa,\, 
    \vecx_H \subseteq \vecu_H,\,
    \vecx_H \cap \veca_H = \emptyset
    \right\},\\
        \Ss_B \coloneqq &\left\{ \vecx_L \uplus \vecb_H \subseteq [d] \ \mid \ 
        \vecb_L \uplus \vecb_H \in \Bb,\, 
        \vecx_L \subseteq \vecu_L,\,
        \vecx_L \cap \vecb_L = \emptyset
    \right\}.
    \end{align*}
    We claim that $\Ss_A \cap \Ss_B \neq \emptyset$ if and only if an orthogonal pair exists.
    Assume that there exist orthogonal vectors $\veca_L \uplus
    \veca_H \in \Aa$ and $\vecb_L \uplus \vecb_H \in \Bb$. Then $\veca_L \uplus \vecb_H$ is contained in both $\Ss_A$ and $\Ss_B$.
    For the converse, assume that $\vecx_L \uplus \vecx_H \in \Ss_A \cap \Ss_B$.
    Since $\vecx_L \uplus \vecx_H \in \Ss_A$, there exists $\veca_L \uplus \veca_H
    \in \Aa$ such that $\vecx_L = \veca_L$ and $\vecx_H \cap \veca_H =
    \emptyset$. Similarly, since $\vecx_L \uplus \vecx_H \in \Ss_B$, there exists
    $\vecb_L \uplus \vecb_H \in \Bb$ such that $\vecb_H = \vecx_H$ and $\vecb_L \cap \vecx_L = \emptyset$.
    Hence, $\veca_L \cap \vecb_L = \emptyset$ and $\veca_H \cap \vecb_H = \emptyset$.

    Observe that both $|\Ss_A|$ and $|\Ss_B|$ are at most $2^{d/2} \cdot (|\Aa| + |\Bb|)$, so in $\Ot(2^{d/2} \cdot n)$ time we can construct $\Ss_A$ and $\Ss_B$, and check whether $\Ss_A \cap \Ss_B \neq \emptyset$.
\end{proof}

\section{Algorithm for Orthogonal Vectors ($2$-OV)}\label{sec:2OV}

\twoov*

The algorithm of \cref{thm:2ov} first randomly partitions the universe $\vecu = [d]$
into $\ell$ sets $\vecu_1, \dots, \vecu_\ell$ of equal size and samples random families
$\Cc_i$ of subsets of $\vecu_i$. We will show that if there exists $\veca \in \Aa$ and
$\vecb \in \Bb$ such that $\veca \cap \vecb = \emptyset$, then with sufficiently high
probability the families $\Cc_1, \ldots, \Cc_\ell$ contain a witness of the
existence of $\veca$ and $\vecb$. 
Finally, we introduce auxiliary families $\Ll^i \subset \Cc_i$ and $\Rr^i \subset \Cc_i$ that allow us to efficiently find this witness, if it exists. 

In the following, we assume that there exist $\veca \in \Aa$ and $\vecb \in \Bb$
with $\veca \cap \vecb = \emptyset$. Let $\alpha,\beta \in [0,1]$ be real numbers
such that $|\veca| = \alpha \cdot d$ and $|\vecb| = \beta \cdot d$.  Note that we can guess the values of $\alpha$ and $\beta$ in $\poly(d)$ time. Hence, from now on, we assume that we know the values of $\alpha$ and $\beta$ precisely. Moreover, we can assume without loss of generality that $\alpha \le \beta$ (by swapping $\Aa$ and $\Bb$) and $\alpha+\beta \le 1$ (as otherwise, the answer is trivially negative).

\subsection{Partition of the universe}\label{sec:preprocessing}
Let $\ell$ be a sufficiently large constant (for the sake of presentation we use $\ell=100$).
By padding, we can assume that $d$,$\alpha d$ and $\beta d$ are multiples of $\ell$. Consider a random
partition of the universe $\vecu = [d]$ into $\ell$ parts
$\vecu_1,\ldots,\vecu_\ell$ of equal cardinality such that $[d] = \vecu_1 \sqcup
\ldots \sqcup \vecu_\ell$ and $|\vecu_i| = d/\ell$ for every $i \in [\ell]$.
We observe that with sufficiently high probability, the solution is partitioned
equally among $\vecu_i$.

\begin{observation}\label{obs:support}
    Let $\veca \in \Aa$ and $\vecb \in \Bb$ with
    $|\veca| = \alpha d$ and $|\vecb| = \beta d$.
    With probability at least $d^{-\Oh(\ell)}$, for every $i \in [\ell]$ it holds that:
	\begin{equation*}\label{eq:support}
        |\vecu_i \cap \veca| =  \frac{\alpha d}{\ell} \text{\quad and \quad }
        |\vecu_i
        \cap \vecb| = \frac{\beta d}{\ell}
        .
    \end{equation*}
\end{observation}
\begin{proof}
    We focus on bounding the probability that $|\vecu_i \cap \veca| = \frac{\alpha
    d}{\ell}$ for every $i \in [\ell]$, as the reasoning for $\vecb$ is symmetric. 
    For fixed integers $k_1, \ldots, k_\ell \geq 0$, 
    the number of partitions of $\vecu$ into $\ell$ equal size sets such that
    $|\vecu_i
    \cap \veca| = k_i$ for every $i \in [\ell]$ is 
    \begin{displaymath} 
		\binom{|\veca|}{k_1,\ldots,k_\ell} \cdot \binom{|\vecu \setminus \veca|}{d / \ell - k_1,\ldots,d / \ell - k_\ell}.
	\end{displaymath}
    Since $|\veca| = \alpha d$, the number of partitions of $\vecu$ into $\ell$
    equal size sets such that $|\vecu_i \cap \veca| = \frac{\alpha d}{\ell}$ for every $i \in [\ell]$ is 
	\begin{displaymath}
		\binom{\alpha d}{\alpha d/\ell,\ldots,\alpha d/\ell} \cdot \binom{(1-\alpha) d}{(1-\alpha) d/\ell,\ldots,(1-\alpha)d/\ell},
	\end{displaymath}
	which can be lower bounded using \cref{lem:multvsentropy} by
	\begin{displaymath}
		2^{h(1/\ell,\ldots,1/\ell) d} \cdot d^{-\Oh(\ell)} = 2^{(\log_2{\ell}) \cdot d} \cdot d^{-\Oh(\ell)} = \ell^d  d^{-\Oh(\ell)}.
	\end{displaymath}
	On the other hand, the number of partitions of $\vecu$ into $\ell$ equal size sets is at most $\ell^{d}$. Hence
	\begin{displaymath}
		\Prob{|\vecu_i \cap \veca| = \frac{\alpha d}{\ell} \text{ for every } i \in [\ell]} \ge 
		\ell^d d^{-\Oh(\ell)} \cdot \ell^{-d} = d^{-\Oh(\ell)}.
		\qedhere
	\end{displaymath}
\end{proof}

\subsection{Certificate of orthogonality}\label{sec:cert_orthogonality}

Let $\lambda = \lambda(\alpha,\beta) \in [0,1]$ and $\kappa = \kappa(\alpha,\beta)$ be parameters to be tuned later that only depend on $\alpha$ and $\beta$.
For every $i \in [\ell]$, we randomly draw a family $\Cc_i$ of $2^{\kappa |\vecu_i|} \cdot d^c$
subsets of $\vecu_i$ of size $\lambda |\vecu_i|$, for some sufficiently large constant $c$, i.e.\ we let 
\begin{equation}\label{eq:candidates}
	\Cc_i \subseteq_R \binom{\vecu_i}{\lambda \cdot |\vecu_i|} \text{ such that
    } |\Cc_i| = 2^{\kappa \cdot |\vecu_i|} \cdot d^c.
\end{equation}
% Sets $C \in \Cc_i$ are called \emph{candidate sets}.

\begin{definition}[Certificate of Orthogonality]
    A tuple $(\vecc_1,\ldots,\vecc_\ell) \in \Cc_1 \times \cdots \times \Cc_\ell$ is a \emph{certificate of orthogonality} if 
    \begin{enumerate}[label=(\alph*)]
        \item There exists $\veca \in \Aa$ such that for every $i \in [\ell]$ it
            holds that $\veca \cap \vecu_i \subseteq \vecc_i$, and\label{em:subs}
        \item There exists $\vecb \in \Bb$ such that for every $i \in [\ell]$ it
            holds that $\vecb \cap \vecc_i = \emptyset$.\label{em:disj}
    \end{enumerate}
\end{definition}

Clearly, if there exists a certificate of orthogonality, then there exists
$\veca \in \Aa$ and $\vecb \in \Bb$ with $\veca \cap \vecb = \emptyset$. However, the converse is
not necessarily true. We prove that for sufficiently many sets in each $\Cc_i$
(i.e. large enough parameter $\kappa$), if there exists $\veca \in \Aa$ and $\vecb \in \Bb$
with $\veca \cap \vecb = \emptyset$, then with sufficiently high probability there exists a certificate of orthogonality in $(\Cc_1, \dots, \Cc_\ell)$. 

\begin{lemma}\label{lem:candidates}
    Assume that there exist disjoint $\veca \in \Aa$ and $\vecb \in \Bb$ with
    $|\veca| = \alpha d$ and $|\vecb| = \beta d$. If
    \begin{align*}
        \lambda & \in [\alpha, 1 - \beta] \text{ and } \\ 
        \kappa & \ge h(\lambda) - (1-\alpha-\beta) \cdot h\left(\frac{\lambda - \alpha}{1-\alpha-\beta}\right),
    \end{align*}
    then with probability at least $d^{-\Oh(\ell)}$ there exists a certificate of orthogonality in $\Cc_1\times \cdots \times \Cc_\ell$.
\end{lemma}
\begin{proof}
    We condition on the event that $|\veca \cap \vecu_i| = \alpha |\vecu_i|$ and
    $|\vecb \cap \vecu_i| = \beta |\vecu_i|$ for all $i \in [\ell]$. 
    By \cref{obs:support}, this happens with probability at least $d^{-\Oh(\ell)}$. 

    Fix $i \in [\ell]$. First, we bound the probability that a set $\vecc$ sampled uniformly at random from  $\binom{\vecu_i}{\lambda |\vecu_i|}$ satisfies $\veca \cap
    \vecu_i
    \subseteq \vecc$ and $\vecb \cap \vecc = \emptyset$. 
    The number of sets $\vecc \subseteq  \binom{\vecu_i}{\lambda |\vecu_i|}$ that satisfy both properties is
    \begin{displaymath}
        \binom{(1-\beta-\alpha) |\vecu_i|}{(\lambda - \alpha) |\vecu_i|}
        .
    \end{displaymath}
    Hence, since $|\vecu_i| = d /\ell$ and using \cref{eq:binom_apx}, the probability
    that $\vecc \subseteq_R \binom{\vecu_i}{\lambda |\vecu_i|}$
    satisfies properties both properties is at least
    \begin{displaymath}
		\binom{(1-\alpha-\beta) |\vecu_i|}{(\lambda - \alpha) |\vecu_i|} \cdot
        \binom{|\vecu_i|}{\lambda |\vecu_i|}^{-1}
        \ge
        2^{\left((1-\alpha-\beta) \cdot h\left(\frac{\lambda - \alpha}{1-\alpha-\beta}\right) - h(\lambda) \right)\cdot d / \ell } \cdot d^{-\Oh(1)} 
    \end{displaymath}
    Hence if $\kappa \ge h(\lambda) - (1-\alpha-\beta) \cdot h\left(\frac{\lambda - \alpha}{1-\alpha-\beta}\right)$, then $|\Cc_i| \geq 1/p$ for $p \coloneq 2^{\left((1-\alpha-\beta) \cdot h\left(\frac{\lambda - \alpha}{1-\alpha-\beta}\right) - h(\lambda) \right)\cdot d / \ell}~\cdot~d^{-\Oh(1)}$. Therefore, the probability that
    there exists $\vecc\in \Cc_i$ satisfying both properties is at least $1 - (1-p)^{1/p} \ge 1-e^{-1}$.
    
    Finally, the probability that for every $i \in [\ell]$ there exists $\vecc \in \Cc_i$ such that \ref{em:subs} and \ref{em:disj} hold is at least $(1-2^{-1})^\ell$ which is larger than success probability of~\cref{obs:support} which is $d^{-\Oh(\ell)}$.
    % 
    % \begin{displaymath}
    %     (1-e^{-1})^\ell \ge d^{-\Oh(\ell)}
    % \end{displaymath}
\end{proof}

\subsection{Auxiliary Families}
In order to efficiently check if there exists a certificate of orthogonality in $\Cc_1 \times \ldots \times \Cc_\ell$, we introduce auxiliary families. 
For every $i \in [\ell]$ and every set $\vecx \subseteq \vecu_i$, define the following two families of sets of $\Cc_i$:
\begin{align}\label{eq:data-structures}
	\Ll^{(i)}(\vecx) \coloneqq \{ \vecc \in \Cc_i \mid \vecx \subseteq \vecc\} \text{ and }
	\Rr^{(i)}(\vecx) \coloneqq \{ \vecc \in \Cc_i \mid \vecx \cap \vecc = \emptyset\}.
\end{align}

\begin{claim}\label{clm:data-structures}
	The families $\Ll^{(i)}(\vecx)$ and $\Rr^{(i)}(\vecx)$ for every $i \in [\ell]$ and
    every $\vecx \subseteq \vecu_i$ can be computed in $\Oh(\ell \cdot 2^{2d/\ell})$ time.
\end{claim}
\begin{proof}
    Since, for every $i \in [\ell]$, the number of subsets $\vecx \subseteq \vecu_i$ is  $2^{d /
    \ell}$, and the number of sets $\vecc$ in $\Cc_i \subset 2^{\vecu_i}$ is at most
    $2^{d / \ell}$, the families of sets can be constructed by iterating over
    all $i \in [\ell]$, all $\vecx \subset \vecu_i$ and all $\vecc \in \Cc_i$ in time at most $\Oh(\ell \cdot 2^{d/\ell} \cdot 2^{d / \ell})$.
\end{proof}

\begin{lemma}\label{lem:ca}
    For a fixed $\veca \in \Aa$ with $|\veca| = \alpha d$, with probability at least $d^{-\Oh(\ell)}$ it holds that:
    \begin{displaymath}
		|\Ll^{(1)} (\veca \cap \vecu_1)| \cdot |\Ll^{(2)} (\veca \cap \vecu_2)| \cdots
        |\Ll^{(\ell)}(\veca \cap \vecu_\ell)| \le 2^{c_A \cdot d} \cdot d^{\Oh(\ell)},
    \end{displaymath}
    where $c_A \coloneq (1-\alpha) \cdot
    h\left(\frac{\lambda-\alpha}{1-\alpha}\right) - h(\lambda) + \kappa$.
\end{lemma}
\begin{proof}
    We condition on the event that $|\veca \cap \vecu_i| = \alpha |\vecu_i|$ for all $i \in [\ell]$, which happens with probability at least $d^{-\Oh(\ell)}$ by \cref{obs:support}.
	% We assume that~\eqref{eq:support} holds because by~\cref{obs:support} the probability of that event is high enough.
    We first prove that for any $i \in [\ell]$ and any $\vecx \subseteq \vecu_i$
    with $|\vecx| = \alpha \cdot |\vecu_i|$ it holds that
	\begin{align}\label{eq:lemca}
		\Ex{|\Ll^{(i)}(\vecx)|} \le 2^{c_A \cdot |\vecu_i|} \cdot d^{\Oh(1)}.
    \end{align}

	By definition, $|\Ll^{(i)}(\vecx)|$ is the number of sets $\vecc \in
    \Cc_i$ such that $\vecx \subseteq \vecc$. Since sets $\vecc \in \Cc_i$ are
    subsets of $\vecu_i$ drawn uniformly at random such that $|\vecc| = \lambda
    |\vecu_i|$, by \cref{eq:binom_apx}, the probability that $\vecx \subseteq \vecc$ is
    \begin{displaymath}
        \Prob{\vecx \subseteq \vecc} = \binom{(1-\alpha)|\vecu_i|}{(\lambda -
		\alpha)|\vecu_i|}\cdot\binom{|\vecu_i|}{\lambda |\vecu_i|}^{-1}
        \le 2^{\left((1-\alpha)h\left(\frac{\lambda-\alpha}{1-\alpha}\right) -
        h(\lambda)\right) \cdot |\vecu_i|}
		.
    \end{displaymath}
    % \anita{Not sure about this computation}
    By the linearity of expectation, this means that:
    \begin{displaymath}
		\Ex{|\Ll^{(i)}(\vecx)|} = |\Cc_i| \cdot \Prob{\vecx \subseteq \vecc} \le
        2^{\left((1-\alpha)h\left(\frac{\lambda-\alpha}{1-\alpha}\right) -
        h(\lambda)\right) \cdot |\vecu_i|} \cdot 2^{\kappa \cdot |\vecu_i|} \cdot d^{\Oh(1)},
    \end{displaymath}
	which establishes~\eqref{eq:lemca}.
	
	Now notice that for any $i\in [\ell]$, by Markov inequality we have that
	\begin{displaymath}
		\Prob{|\Ll^{(i)}(\vecx)| \ge 2 \Ex{|\Ll^{(i)}(\vecx)|}} \le 1/2,
	\end{displaymath}
	therefore the probability that we have $|\Ll^{(i)}(\veca\cap \vecu_i)| < 2
    \Ex{|\Ll^{(i)}(\veca\cap \vecu_i)|}$ for all $i \in [\ell]$ is at least $1/2^\ell \ge d^{-\Oh(\ell)}$. This directly implies that with probability at least $d^{-\Oh(\ell)}$ we have
	    \begin{displaymath}
		|\Ll^{(1)} (\veca \cap \vecu_1)| \cdot |\Ll^{(2)} (\veca \cap \vecu_2)| \cdots
        |\Ll^{(\ell)}(\veca \cap \vecu_\ell)| \le 2^{c_A \cdot d} \cdot
        d^{\Oh(\ell)}.\qedhere
    \end{displaymath}
\end{proof}
 
Similarly, one can prove:
\begin{lemma}\label{lem:cb}
    For a fixed $\vecb \in \Bb$ with $|\vecb| = \beta d$, with probability at least $d^{-\Oh(\ell)}$ it holds that:
    \begin{displaymath}
		|\Rr^{(1)}(\vecb \cap \vecu_1)| \cdot |\Rr^{(2)}(\vecb \cap \vecu_2)| \cdots
        |\Rr^{(\ell)}(\vecb \cap \vecu_\ell)| \le 2^{c_B \cdot d} \cdot d^{\Oh(\ell)},
    \end{displaymath}
    where $c_B \coloneq (1-\beta) \cdot
    h\left(\frac{1-\lambda-\beta}{1-\beta}\right) - h(\lambda) + \kappa$.
\end{lemma}
\begin{proof}
    The proof is similar as for \cref{lem:ca} by observing that, for every $i \in [\ell]$
    $$
    \Prob{\vecx \cap \vecc = \emptyset} = \binom{(1 -
    \beta)|\vecu_i|}{\lambda|\vecu_i|}\cdot\binom{|\vecu_i|}{\lambda
    |\vecu_i|}^{-1},
    $$
    for sets $\vecx, \vecc \subset \vecu_i$ with $|\vecx| = \beta |\vecu_i|$ and
    $|\vecc| = \lambda |\vecu_i|$.
\end{proof}

% \karol{Do we need a proof of that?}

\subsection{Algorithm}

Now, with $\Ll^{(i)}$ and $\Rr^{(i)}$ in hand, we can present the algorithm given in pseudocode in~\cref{alg:2ov}. First, we construct a set
$\textsf{Candidates}$ of tuples in $\Cc_1 \times \cdots \times \Cc_\ell$ that
satisfy property~\ref{em:subs}. We do this by iterating over every set $\veca
\in \Aa$ and for such $\veca$ we iterate over every $(\vecc_1,\ldots,\vecc_\ell) \in
\Ll^{(1)}(\veca
\cap \vecu_1) \times \cdots \times \Ll^{(\ell)}(\veca \cap \vecu_\ell)$ and add it to the set
$\textsf{Candidates}$. 

Clearly, after this step, every set that satisfies~\ref{em:subs} is in \textsf{Candidates}. Now, our
goal is to decide if there exists at least one tuple in \textsf{Candidates} that additionally satisfies~\ref{em:disj}.
To do so, we iterate over every $\vecb \in \Bb$ and then over every 
$(\vecc_1,\ldots,\vecc_\ell) \in \Rr^{(1)}(\vecb \cap \vecu_1) \times \cdots \times
\Rr^{(\ell)}(\vecb \cap \vecu_\ell)$. If $(\vecc_1,\ldots,\vecc_\ell) \in \textsf{Candidates}$, then we know that there exists a certificate of orthogonality and immediately report that there exists an orthogonal pair.
If none of the loops reported yes, at the end of the algorithm we report that there is no pair of orthogonal vectors.

\begin{algorithm}[ht!]
    \DontPrintSemicolon
    \textbf{function} $\textsf{OV}(\Aa,\Bb)$:\\
    Guess $\alpha d,\beta d \in [d]$ \tcp*{Only $\Oh(d^{2})$ possibilities.}
    Randomly partition $[d] = \vecu_1 \sqcup \ldots \sqcup \vecu_\ell$ into equal size sets \tcp*{see~Sec~\ref{sec:preprocessing}}
    Draw families $\Cc_1,\ldots,\Cc_\ell$ as defined in \eqref{eq:candidates} \\
	Construct $\Ll^{(1)},\ldots,\Ll^{(\ell)},\Rr^{(1)},\ldots,\Rr^{(\ell)}$ as defined in \eqref{eq:data-structures}\\
    Let $\textsf{Candidates} = \{\}$\\
	\ForEach{$\veca \in \Aa$}{
		\ForEach(\tcp*[f]{break if more than $2^{c_A \cdot
        d}$}){$(\vecc_1,\ldots,\vecc_\ell) \in \Ll^{(1)}(\veca \cap \vecu_1) \times \cdots \times
		\Ll^{(\ell)}(\veca \cap \vecu_\ell)$\label{ln:itera}}{
            Add $(\vecc_1,\ldots,\vecc_\ell)$ to \textsf{Candidates}.
        }
    }
    \ForEach{$\vecb \in \Bb$}{
		\ForEach(\tcp*[f]{break if more than $2^{c_B \cdot
        d}$}){$(\vecc_1,\ldots,\vecc_\ell) \in \Rr^{(1)}(\vecb \cap \vecu_1) \times \cdots \times
		\Rr^{(\ell)}(\vecb \cap \vecu_\ell)$\label{ln:iterb}}{
            \If{$(\vecc_1,\ldots,\vecc_\ell) \in \normalfont{\textsf{Candidates}}$}{\Return Exist Orthogonal Pair}
        }
    }
    \Return No Orthogonal Pair
	\caption{Pseudocode of the algorithm of~\cref{lem:2ov-with-constants}.}
    \label{alg:2ov}
\end{algorithm}

Finally, to simplify the analysis of \cref{alg:2ov}, we stop iterating over every candidate set in Line~\ref{ln:itera} and Line \ref{ln:iterb} if the number of iterations exceeds respectively $2^{c_A \cdot d} \cdot d^{\Oh(\ell)}$ or $2^{c_B \cdot d} \cdot d^{\Oh(\ell)}$ steps, where $c_A$ and $c_B$ are the constants defined in \cref{lem:ca,lem:cb}.
When we break these loops, we do not add any sets to $\textsf{Candidates}$ and immediately continue the loops.

% We are almost done with the description of~\cref{alg:2ov}, except for one
% tweak that we do to simplify the running time analysis. Namely, when we iterate over every candidate set in Line~\ref{ln:itera}, we stop if the number of iterations exceeds $2^{c_A \cdot d} \cdot d^{\Oh(\ell)}$ (see~\cref{lem:ca} for the definition of constant $c_A$). Similarly, we stop the execution of the loop in line~\ref{ln:iterb} if it exceeds $2^{c_B \cdot d} \cdot d^{\Oh(\ell)}$ steps (see~\cref{lem:cb} for the definition of $c_B$). When we break these loops, we do not add any sets to $\textsf{Candidates}$ and immediately continue the loops.

This concludes the description of~\cref{alg:2ov}. The following lemma certifies
its correctness and bounds its running time in terms of the parameters $\alpha,
\beta, \lambda(\alpha, \beta)$ and $\kappa(\alpha, \beta)$. We then provide an analytical analysis of the running time by optimizing $\alpha$ and $\beta$.

\begin{lemma}\label{lem:2ov-with-constants}
    Given $\Aa,\Bb \subseteq 2^{[d]}$, there exists an 
	\begin{displaymath}
		\left(2^{c_A \cdot d} |\Aa| + 2^{c_B \cdot d} |\Bb| + \ell \cdot 2^{2d/\ell}\right) \cdot d^{\Oh(\ell)}
	\end{displaymath}
	time randomized, one-sided error algorithm with $1 - 2^{-d^{\Omega(1)}}$ probability of success that decides if there
    exist $\veca \in \Aa$ and $\vecb \in \Bb$ such that $\veca \cap \vecb =
    \emptyset$, where $c_A$ and $c_B$ are as in \cref{lem:ca,lem:cb}.
\end{lemma}
\begin{proof}
	We repeat~\cref{alg:2ov} $d^{\Oh(\ell)}$ times and return true iff any execution returns true.
	Consider~\cref{alg:2ov}. The preprocessing in~\cref{sec:preprocessing} takes
    $\Oh(d)$ time.  By~\cref{clm:data-structures}, the construction of data
	structures $\Ll^{(1)},\ldots,\Ll^{(\ell)},\Rr^{(1)},\ldots,\Rr^{(\ell)}$ takes $\Oh(\ell \cdot 2^{2d/\ell})$
	time.  Finally, the for loops in~\cref{alg:2ov}
	take in total $(2^{c_A \cdot d}|\Aa| + 2^{c_B \cdot d} |\Bb|) \cdot d^{\Oh(\ell)}$ time.
	This concludes the running time analysis. It remains to prove the correctness.

	Note that if there do not exist any disjoint pairs in $\Aa$ and $\Bb$,
	then~\cref{alg:2ov} never returns true. Hence, to bound the probability of
	false-negative, assume that $\veca \in \Aa$ and $\vecb \in \Bb$ are orthogonal.
    By~\cref{lem:candidates}, with probability at least $d^{-\Oh(\ell)}$, there exists a certificate of orthogonality in $\Cc_1 \times \dots \times \Cc_\ell$ for $\veca$ and $\vecb$. 
	By~\cref{lem:ca,lem:cb}, with probability at least $d^{-\Oh(\ell)}$, the loop in Line~\ref{ln:itera} corresponding to $\veca$ and the loop in Line~\ref{ln:iterb} corresponding to $\vecb$ are fully executed.
    % the number of
	% iterations of the loop in Line~\ref{ln:itera} and the loop in Line~\ref{ln:iterb} is respectively at most $2^{c_A \cdot d} \cdot d^{\Oh(\ell)}$ and at most $2^{c_B \cdot d} \cdot d^{\Oh(\ell)}$. Hence, with probability at least $ d^{-\Oh(\ell)}$, the
    % loops for $\veca \in \Aa$ and $\vecb\in \Bb$ will be fully executed.
    % 
    Observe that by~\cref{lem:ca}, with probability at least $d^{-\Oh(\ell)}$, the loop in Line~\ref{ln:itera} corresponding to $\veca$ is fully executed. 
	If that happens, the set $\textsf{Candidates}$ contains every tuple
	$(\vecc_1,\ldots,\vecc_\ell) \in \Ll^{(1)}(\veca \cap \vecu_1) \times \cdots \times
    \Ll^{(\ell)}(\veca \cap
	\vecu_\ell)$. In particular, it contains the certificate of orthogonality for $\veca$ and $\vecb$.
    % Therefore, by~\cref{lem:candidates}, we know that set
	% $\textsf{Candidates}$ contains a certificate of orthogonality for $\veca$
    % and $\vecb$. 
    Moreover, by~\cref{lem:cb}, with probability at least $d^{-\Oh(\ell)}$, the loop in Line~\ref{ln:iterb} corresponding to $\vecb$ is also fully executed. Hence this certificate will be detected and the algorithm returns true. 

	To conclude, if $\veca \in \Aa$ and $\vecb \in \Bb$ are orthogonal, \cref{alg:2ov}
	returns true with probability $d^{-\Oh(\ell)}$. To guarantee probability of success $1-2^{-d^{\Omega(1)}}$, it suffices to repeat \cref{alg:2ov} $d^{\Oh(\ell)}$
	many times.
\end{proof}

\begin{proof}[Proof of~\cref{thm:2ov}]
    Assume that there exist $\veca \in \Aa$ and $\vecb \in \Bb$ such that $\veca
    \cap \vecb = \empty$ and let $|\veca| = \alpha d$ and $|\vecb| = \beta d$ for some $\alpha, \beta \in [0, 1]$.
	By~\cref{lem:2ov-with-constants}, one can find the pair $(\veca, \vecb)$ in time
	\begin{displaymath}
		\Ot\left(2^{c_A \cdot d} |\Aa| + 2^{c_B \cdot d} |\Bb| + \ell \cdot
        2^{2d/\ell}\right),
	\end{displaymath}
    where $\lambda$ and $\kappa$ are parameters depending on $\alpha$ and $\beta$ to
    be optimized, and $c_A$ and $c_B$ are constants depending on $\alpha$,
    $\beta$, $\lambda$ and $\kappa$. 
	Note that we have selected $\ell = 100$, so the last term is upper bounded
    by $\Ot(1.25^d \cdot n)$. Therefore, it suffices to prove that for any
    $\alpha,\beta$, by properly selecting $\lambda$ and $\kappa$, we can bound $2^{c_A d} + 2^{c_B d}$ by $\Ot(1.25^d)$.
	
    First, observe that the constants $c_A$ and $c_B$ defined in
    \cref{lem:ca,lem:cb} grow with the parameter $\kappa$. On the other hand, by
    \cref{lem:candidates}, we need $\kappa \geq h(\lambda) - (1-\alpha - \beta) \cdot
    h(\frac{\lambda - \alpha}{1 - \alpha - \beta})$ for the correctness analysis
    in \cref{lem:2ov-with-constants} to hold. Hence to minimise $2^{c_A d} +
    2^{c_B d}$, we choose $\kappa$ to be minimal, i.e.~ we set $\kappa \coloneq h(\lambda) - (1-\alpha - \beta) \cdot h(\frac{\lambda - \alpha}{1 - \alpha - \beta})$. Then we have 
    \begin{equation}\label{eq:running_time}
    \begin{aligned}
        c_A &= (1-\alpha) \cdot h\left(\frac{\lambda-\alpha}{1-\alpha}\right) -
        (1 - \alpha - \beta) \cdot h\left(\frac{\lambda - \alpha}{1 - \alpha -
        \beta}\right), \\ 
        c_B &= (1-\beta) \cdot h\left(\frac{1-\lambda-\beta}{1-\beta}\right) -
        (1 - \alpha - \beta) \cdot h\left(\frac{\lambda - \alpha}{1 - \alpha -
        \beta}\right),
    \end{aligned}
    \end{equation}
    and using \Cref{lem:multvsentropy}, we get the following bound
    \begin{align*}
        2^{c_A d} + 2^{c_B d} 
        &\leq \Ot\left(\left( \binom{(1-\alpha)d}{(\lambda - \alpha)d} +
        \binom{(1-\beta)d}{(1-\lambda-\beta )d}\right)\cdot
    \binom{(1-\alpha-\beta)d}{(\lambda - \alpha)d}^{-1} \right).
    \intertext{So by setting $\lambda \coloneq (1 + \alpha - \beta) / 2$ and noticing that $\binom{(1-\alpha-\beta)d}{(1 - \alpha - \beta)d/2}^{-1} \le \Ot(2^{-(1-\alpha-\beta)d})$, we can further bound }
        2^{c_A d} + 2^{c_B d} 
        &\leq\Ot\left(\left( \binom{(1-\alpha)d}{(1- \alpha - \beta)d/2} + \binom{(1-\beta)d}{(1-\alpha-\beta )d/2}\right)\cdot \binom{(1-\alpha-\beta)d}{(1- \alpha - \beta)d/2}^{-1} \right) \\
        &\leq \Ot\left(\left( \binom{(1-\alpha)d}{(1- \alpha - \beta)d/2} +
        \binom{(1-\beta)d}{(1-\alpha-\beta )d/2}\right)\cdot
    2^{-(1-\alpha-\beta)d}  \right).
    \intertext{
    Notice that both binomials select $(1-\alpha-\beta)d/2$ elements and recall that we assume $\alpha \leq \beta$. Hence the left term dominates the right term and by using \cref{lem:multvsentropy} we bound the expression}
        2^{c_A d} + 2^{c_B d} 
        &\leq \Ot\left( \binom{(1-\alpha)d}{(1- \alpha - \beta)d/2} \cdot 2^{-(1-\alpha-\beta)d}  \right) \\
        &\leq \Ot\left( 2^{h\left(\frac{(1 - \alpha - \beta)/2}{1-\alpha}\right) \cdot  (1-\alpha)d}\cdot 2^{-(1-\alpha-\beta)d}  \right) \\
        &= \Ot\left( 2^{h\left(\frac{1}{2} - \frac{\beta}{2(1-\alpha)}\right) \cdot  (1-\alpha)d -(1-\alpha-\beta)d}  \right) 
        % \\ &
        = \Ot\left( 2^{\left(h\left(\frac{1}{2} -
        \frac{\beta}{2(1-\alpha)}\right)  -1\right) \cdot (1-\alpha)d + \beta d} \right).
    \end{align*}
    The term involving $\alpha$, that is $\left(h\left(\frac{1}{2} - \frac{\beta}{2(1-\alpha)}\right)  -1\right) \cdot (1-\alpha)d$, is maximized when $\alpha = 0$. This setting allows $\beta$ to take any value in $[0, 1]$. Hence the above term is maximized when $\alpha = 0$ and $\beta \in [0, 1]$ is maximizing $h(1/2 + \beta/2) - 1 + \beta$. 
    The derivative is $(\log_2(4-4\beta) - \log_2(\beta+1))/2$, therefore we conclude that this expression has a single local
    maximum at $\beta = 3/5$. In that case, the running time is bounded by:
        \begin{displaymath}
            \Ot\left(2^{(h(1/5) - 1 + 3/5)\cdot d} \cdot n \right) =
            \Ot\left(
                1.25^{d} \cdot n
            \right).\qedhere
        \end{displaymath}
\end{proof}

\begin{remark}
    The computer evaluation suggests that the maximum is obtained for the choice $\alpha =
    \beta = 1/3$, for which the best $\lambda = 1/2$. See~\cref{fig:opt_running_time}. This choice yields the running time:
    \begin{align*}
        \Ot\left(\left( \frac{4}{3} \right)^{d/2}\cdot n\right) \le \Oh(1.16^d \cdot n)
        .
    \end{align*}
\end{remark}

\section{Algorithm for $k$-Orthogonal Vectors ($k$-OV)}\label{sec:kOV}

We now show that for every fixed $k \ge 2$, there exists an algorithm for $k$-OV with $\Oh(2^{(1-\eps_k) \cdot d}\cdot n)$ running time, where $\eps_k>0$ is a constant depending only on $k$. We remark that $\eps_k$ tends to $0$ as $k$ tends to infinity. This is consistent with the lower bound we present in \cref{thm:lowerBound}.

\kov*

To prove \cref{thm:kov}, we will need the following combinatorial algorithm that follows from \cite[Theorem 1]{bjorklund-ov}. 

\begin{lemma}\label{thm:kov_dense}
    Let $k \in \nat$ be a fixed integer. Given $\Aa_1,\ldots,\Aa_k \subseteq
    2^{[d]}$, we can count the number of tuples $(\veca_1, \dots,
    \veca_k) \in \Aa_1 \times \dots \times \Aa_k$ such that $\veca_1 \cap
    \veca_2 \cap \ldots \cap \veca_k = \emptyset$ in time at most $\Oh(d \cdot (|\dc \Aa_1| + \ldots + |\dc \Aa_k| ))$.
\end{lemma}

\begin{proof}
    \newcommand{\Count}{\#\text{OV}}
    This directly follows from \cite[Theorem 1]{bjorklund-ov} extended to $k$ families. For completeness, we detail the proof here.

    Let $\Count$ be the number of solutions for the given instance, i.e. the
    number of tuples $(\veca_1, \dots, \veca_k) \in \Aa_1 \times \dots \times
    \Aa_k$ such that $\veca_1 \cap \dots \cap \veca_k = \emptyset$. We can
    express $\Count$ as follows, where the second equality comes from the fact that any non-empty set has as many subsets of even size as subsets of odd size.
    \begin{align*}
    \Count &= \sum_{\veca_1 \in \Aa_1} \sum_{\veca_2 \in \Aa_2} \dots \sum_{\veca_k
    \in \Aa_k} \iv{\veca_1 \cap \veca_2 \cap \dots \cap \veca_k = \emptyset} \\
    &= \sum_{\veca_1 \in \Aa_1} \sum_{\veca_2 \in \Aa_2} \dots \sum_{\veca_k \in
    \Aa_k} \sum_{\vecx \subset [d]} (-1)^{|\vecx|} \cdot \iv{\vecx \subset \veca_1 \cap \veca_2
    \cap \dots \cap \veca_k} \\ 
    &= \sum_{\vecx \subset [d]} (-1)^{|\vecx|} \sum_{\veca_1 \in \Aa_1} \sum_{\veca_2
    \in \Aa_2} \dots \sum_{\veca_k \in \Aa_k}  \iv{\vecx \subset \veca_1} \cdot
    \iv{\vecx \subset \veca_2} \dots \iv{\vecx \subset \veca_k} \\ 
    &= \sum_{\vecx \in \bigcup_{i = 1}^k \dc \Aa_i} (-1)^{|\vecx|} \cdot f_1
    (\vecx) \cdot f_2(\vecx) \cdots f_k(\vecx)
    \end{align*}
    where $f_i(\vecx) \coloneq \sum_{\veca \in \Aa_i} \iv{\vecx \subset \veca}$ for $i \in [k]$. 
    For each $i \in [k]$, we compute the values of $f_i(\cdot)$ for all $\vecx \in
    \dc \Aa_i$ in time $\Oh(d \cdot |\dc \Aa_i|)$ as follows. For each element
    $j \in \{0, 1, \dots, d\}$ and subset $\vecx \in \dc \Aa_i$ define
    $g_j(\vecx)$ to
    be the number of sets $\veca \in \Aa_i$ such that $\vecx \subset \veca$ and
    $\veca\cap \{1, 2, \dots, j\} = \vecx \cap \{1, 2, \dots, j\}$. In
    particular, $g_0 (\vecx) = f_i(\vecx)$ and $g_d(\vecx) = \iv{\vecx \in A_i}$. By induction on $j$, we can prove that 
    $$
    g_{j-1}(\vecx) = \iv{j \notin \vecx} \cdot g_j(\vecx) + \iv{\vecx \cup \{j\} \in \dc
    \Aa_i} \cdot g_j(\vecx \cup \{j\}).
    $$
    Using this recurrence relation, we can compute $g_j(\vecx)$ for all $\vecx
    \in \dc \Aa_i$ and for all $j \in \{0, 1, \dots, d\}$, and thus in
    particular $f_i(\vecx)$, in time $\Oh(d |\dc \Aa_i|)$. 
    By repeating this for every $i \in [k]$ and using the above formula for
    $\Count$, we can count the number of $k$-OV solutions in time $\Oh(d \cdot (|\dc \Aa_1| + \dots + |\dc \Aa_k|))$. 
% \anita{"If one solution exists, we can also reconstruct it in the same time." should be true but hum how?}
%  \evangelos{I think once we have an overall algorithm, for both dense and sparse instances, we can do it blackbox. Remove half the As, check if there exists a solution still; if yes, continue in the remaining half, if not continue in the removed half. I suggest you just rewrite the theorem to decide if a $k$-tuple exists, without ever promising to return it, just to make everyone's life easier.}
\end{proof}

% \evangelos{This is solving the general case, not just sparse instances; shall we rename sections?}
% \evangelos{How about adding: We now show that for every fixed $k \in \nat$, there exists an algorithm for $k$-OV with running time $2^{(1-\eps_k) \cdot d}\cdot n$, where $\eps_k>0$ is a constant depending only on $k$. We remark that $\eps_k$ tends to $0$ as $k$ tends to infinity. This is consistent with the lower bound we present in \cref{thm:lowerBound}.}

% \begin{theorem}
%     Let $k \in \nat$ be a fixed integer. There exist $\eps_k > 0$, such that for
%     a given $\Aa_1,\ldots,\Aa_k \subseteq
%     \{1,\ldots,d\}$ of total size $n$ there exists an algorithm that runs in
%     time $\Oh(2^{(1-\eps_k) d} \cdot n)$ and can decide if $a_1 \in \Aa_1,\ldots,a_k \in \Aa_k$ such that $a_1
%     \cap a_2 \ldots \cap a_k = \emptyset$ exist.
% \end{theorem}

% \anita{I know that it is the same, but it doesn't *look* the same: could it be confusing that we sometimes say $2^{(1-\epsilon_k)d} \cdot n$ and sometimes $(2-\epsilon_k)^d \cdot n^{\Oh(1)}$ (especially when comparing with the lower bound)}

\begin{proof}[Proof of \cref{thm:kov}]
    Fix $k \geq 3$ and let $0 < \epsilon_k \leq 1/2$ be a parameter depending
    only on $k$ that we define later. Consider a $k$-OV instance $\Aa_1, \dots,
    \Aa_k \subset 2^{[d]}$.
    Note that if the down-closures of $\Aa_1, \dots, \Aa_k$ are all small enough,
    i.e.\ of size at most $\Oh(2^{(1-\eps_k)d} \cdot n)$, then
    \cref{thm:kov_dense} provides the desired running time. If this is not the
    case, then we reduce the $k$-OV instance to a $k'$-OV instance
    for smaller $k'< k$. Since \cref{thm:det-ov} proves the statement for $k=2$ with value $\epsilon_2 = 1/2$, we only need to show that $\epsilon_{k}$ stays positive.

    Assume that there exists a solution $(\veca_1, \dots, \veca_k) \in \Aa_1 \times
    \dots \times \Aa_k$ such that $\veca_1 \cap \dots \cap \veca_k = \emptyset$.
    Let $\alpha_1, \dots, \alpha_k \in [0,1]$ be such that $|\veca_i| = \alpha_i
    d$ for every $i \in [k]$. By enumerating all possible sizes, we can guess
    the values of $\alpha_1, \dots, \alpha_k$ in $\Oh(d^k)$ time. 
    This factor will be negligible compared to $2^{(1-\eps_k)d}$ as $k\ge2$ is a
    fixed constant. After guessing $\alpha_1,\ldots,\alpha_k$, we assume that the
    family $\Aa_i$ contains only sets of cardinality exactly $\alpha_i d$ for every $i \in [k]$. 
    Next, we distinguish between two cases.

    If $\alpha_i \leq (1 - \eps_k)$ for every $i \in [k]$, then the
    size of the down-closure of $\Aa_i$ is $|\dc \Aa_i| = \Oh(2^{\alpha_i d} \cdot
    |\Aa_i|)$. Using the algorithm of \cref{thm:kov_dense}, we solve the $k$-OV instance $\Aa_1, \dots,
    \Aa_k$  in time $\Oh( 2^{(1 -
    \eps_k)d} \cdot n \cdot d)$.  
    
    If there exists $i \in [k]$ such that $\alpha_i > (1 - \eps_k)$, then we recurse as follows. By reordering the families, we can assume that $\alpha_1 \leq \dots \leq\alpha_k$.
    Let $\ell \in [k]$ be the smallest index such that $\alpha_\ell > (1 - \eps_k)$, i.e. $\alpha_1 \leq \dots \leq \alpha_{\ell-1} \leq (1 - \eps_k) < \alpha_\ell \leq \dots \leq \alpha_k$. 
    Note that $(\veca_\ell, \dots, \veca_k) \in \Aa_\ell \times \dots \times \Aa_k$. Furthermore, for every $i \in \{\ell, \dots, k\}$, since $\eps_k \leq 1/2$, we can bound $|\Aa_i| \leq \sum_{j = (1-\eps_k)d}^d \binom{d}{j} \leq \binom{d}{(1- \eps_k)d} \cdot d = \binom{d}{\eps_k d} \cdot d$. By \cref{eq:binom_apx} we thus get $|\Aa_i| \leq 2^{h(\eps_k)d} \cdot d$.
    Hence we can enumerate $\Aa_\ell \times \dots \times \Aa_k$ to guess the
    sets $\veca_\ell, \dots, \veca_k$ in time $\Oh\left(2^{h(\eps_k)d (k - \ell
    + 1)}\right)$. To find the remaining sets $\veca_1, \dots, \veca_{\ell -1}$,
    observe that we can restrict the universe to $\vecr \coloneq \veca_\ell \cap
    \dots \cap \veca_k$ as we have $\veca_1 \cap \dots \cap \veca_k = \veca_1
    \cap \dots \cap \veca_{\ell -1} \cap \vecr = (\veca_1 \cap \vecr) \cap \dots \cap
    (\veca_{\ell-1} \cap \vecr)$. So we recursively solve the  $(\ell - 1)$-OV
    instance $\Aa_1^r, \dots, \Aa_{\ell-1}^r$ where $\Aa_i^r \coloneq \{\veca
    \cap \vecr \ \mid \  \veca \in \Aa_i\}$ for all $i \in \{1, \dots, \ell-1\}$. By induction, this takes time $\Ot(2^{(1- \eps_{\ell-1})d} \cdot n)$.
    In total, we solve the $k$-OV instance $\Aa_1, \dots, \Aa_k$ in time
    % \evangelos{Is that $\Ot(2^{h(\eps_k) d (k - \ell +1)} \cdot 2^{(1- \eps_{\ell-1})d} \cdot n)$? }
    $\Ot(2^{h(\eps_k) d (k - \ell +1)} \cdot 2^{(1- \eps_{\ell-1})d} \cdot n) = \Ot(2^{(1 - \eps_k)d} \cdot n)$ for $\eps_k \leq \eps_{\ell-1} - h(\eps_k) \cdot (k - \ell + 1)$. 
    Note that since $\eps_k \leq 1/2$, the function $\eps_k \mapsto \eps_k + h(\eps_k) \cdot (k-\ell +1)$ is increasing.
    {Furthermore, $\lim_{\eps_k \rightarrow 0^+} \eps_k + h(\eps_k) \cdot (k-\ell +1) = 0$.} 
    Since $\eps_2 > 0$ by \cref{thm:det-ov}, we can set $\eps_k$ to a positive value satisfying $\eps_k \leq \min \{1/2, \eps_{\ell-1} - h(\eps_k) \cdot (k - \ell + 1)\}$. 
    {Achieving running time $\Oh( 2^{(1 - \eps_k)d} \cdot n)$, instead of $\Ot( 2^{(1 - \eps_k)d} \cdot n)$ follows directly by scaling all $\eps_k$ (e.g.\ by $0.99$).}
\end{proof}

\begin{figure}
    \centering
    \includegraphics[width=.5\textwidth]{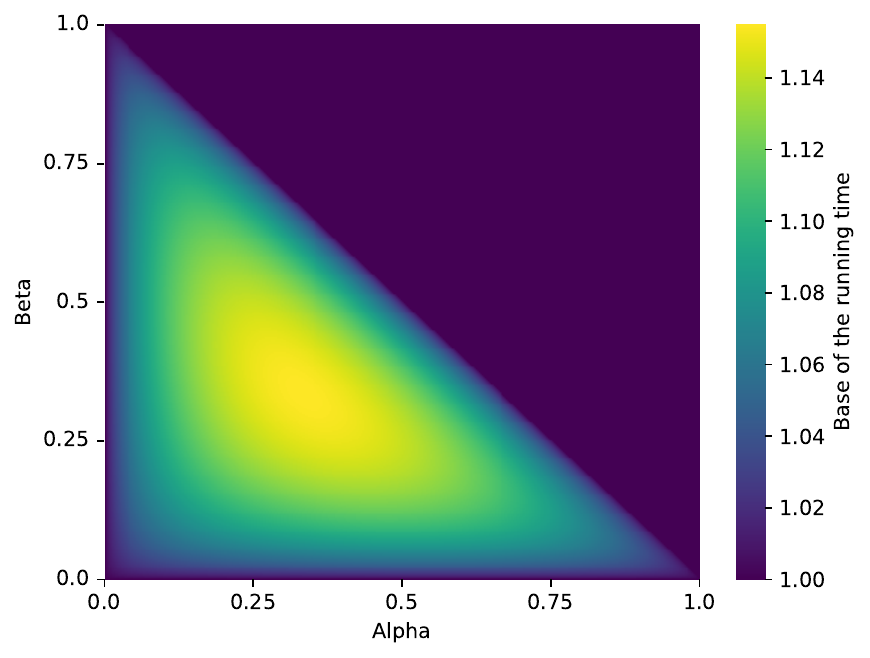}
    \caption{
    For different values of $\alpha, \beta \in [0, 1]$, we optimize the value of $\lambda(\alpha, \beta)$ to minimise the running time bound $\Ot(x^d \cdot n)$ of our algorithm, where the base of the running time is $x = \max\{2^{c_A}, 2^{c_B}\}$, and $c_A$ and $c_B$ are as in~\eqref{eq:running_time}. The running time is maximized for $\alpha = \beta = 1/3$, for which the best found value for $\lambda$ is $\lambda = 1/2$, resulting in a running time base $x \leq 1.16$.
    % The worst-case running time is obtained for $\alpha = \beta = 1/3$, for which the best value of $\lambda$ is $\lambda = 1/2$. 
    }\label{fig:opt_running_time}
\end{figure}

\begin{comment}
\section{Improvements Assuming the Asymptotic Rank Conjecture}

\begin{theorem}
    Assuming the asymptotic rank conjecture, $k$-OV can be solved in $\Oh(2^{0.9d} +
    n)$ time.
\end{theorem}
\karol{TODO: introduce the asymptotic rank conjecture}

\begin{lemma}
    $k$-Orthogonal Vectors is equivalent to $k$-Set Cover
\end{lemma}

\begin{proof}
    [$k$-Set Cover reduces to $k$-Orthogonal Vectors]
    Let $S$ be the $k$-Set Cover instance, i.e. a collection of sets in $[d]$. Observe that there exists $s_1, \dots, s_k \in S$ with $s_1 \cup \dots \cup s_k = [d]$ if and only if $([d] \setminus s_1) \cap \dots ([d] \setminus s_k) = \emptyset$. Hence, there is a $k$-Set Cover in $S$ if and only if there are $k$ orthogonal vectors in $\Aa_1, \dots, \Aa_k$, where $\Aa_i \coloneq \{[d] \setminus s \ \mid \ s \in S\}$.

    [$k$-Orthogonal Vectors reduces to $k$-Set Cover]
    \karol{TODO:intermediate step}
    \anita{tbh I don't remember this step at all}
\end{proof}

\section{Lower Bound for $4$-Cycle}

\begin{theorem}
    If $4$-cycle in directed graphs can be detected in $\Oh(m^{1.1})$ time then
    $2$-OV can be decided in $\Oh(2^{0.2d} n + 2^{0.25d})$ time.
\end{theorem}
\begin{proof}
    \karol{TODO: in the algorithm for $2$-OV you basically split universe into
    two. Fill in the correct numbers}
\end{proof}

\end{comment}

% \appendix

\section{Faster $k$-Orthogonal Vectors under the Asymptotic Rank Conjecture}\label{app:SetCover}

In the Set Cover problem, we are given a universe $\vecu$, a set family $\Ff$
and a positive integer $t$. The goal is to determine if there exist $t$ sets in
$\Ff$ whose union is $\vecu$. We say that a set family $\Ff$ is
$\delta$-bounded for some $\delta \in (0,1)$ if no set in $\Ff$ is larger than
$\delta|\vecu|$. Bj\"orklund et al.~\cite{radu25} showed that under the Asymptotic Rank
Conjecture, Set Cover instances with $\delta$-bounded set families can be solved
faster than $2^{|\vecu|}$ time.

\begin{theorem}[Theorem 1.2~\cite{radu25}]\label{thm:arc-set-cover}
    Let $\delta \in (0,1/4)$ be a fixed constant. Assuming the Asymptotic Rank Conjecture, 
    there exist $c > 0$ and $\eps > 0$, such that
    any Set Cover instance $(\vecu,\Ff,t)$
where $\Ff$ is $\delta$-bounded, can be solved deterministically in $\Oh((2-\eps)^{|\vecu|} \cdot |\vecu|^c)$ time.
\end{theorem}

Notice, that the running time of the above theorem does not depend on $|\Ff|$ (which is upper bounded by $\binom{|\vecu|}{|\vecu| / 4} \le (\frac{e |\vecu|}{|\vecu| / 4})^{|\vecu|/4} < 1.82^{|\vecu|}$).
In particular, it works even if $|\Ff| = \binom{|\vecu|}{\delta |\vecu|}$. We exploit that fact to prove the following result, which eliminates the necessity for $\delta$-boundedness.

\begin{corollary}\label{cor:arc-set-cover}
    Assuming the Asymptotic Rank Conjecture, there exist $c > 0$ and $\eps > 0$, such that
    any Set Cover instance $(\vecu,\Ff,t)$ can be solved in $\Oh(2^{(1-\eps)\cdot{|\vecu|}} \cdot (|\vecu| + |\Ff|)^c)$ time.
\end{corollary}
\begin{proof}
    Let $\vecs_1, \ldots, \vecs_t \in \Ff$ be a solution to the given Set Cover instance, and $c',\eps'$ be the constants from \cref{thm:arc-set-cover}.
    Let $\alpha = 1/5$. If $|\vecs_i| \leq \alpha \cdot |\vecu|$ for every $i \in [t]$, then we can solve the Set Cover instance in $\Oh((2-\eps')^{|\vecu|} \cdot |\vecu|^{c'})$ time using~\cref{thm:arc-set-cover}, by considering only the sets of size at most $\alpha \cdot |\vecu|$ in $\Ff$.

    Otherwise, there exists $i \in [t]$ such that $|\vecs_i| > \alpha \cdot |\vecu|$. By guessing $\vecs_i$ among the sets in $\Ff$ (which takes an extra $\Oh(|\Ff|)$ factor in the running time), we can restrict the universe to $\vecu' = \vecu \setminus \vecs_i$ and solve the remaining $(t-1)$-Set Cover instance $(\vecu', \Ff', t-1)$, where $\Ff' = \{\vecs \cap \vecu' \ \mid \ \vecs \in \Ff\}$. As the size of the universe is $|\vecu'| \leq (1-\alpha) \cdot |\vecu|$, we can solve the remaining instance in $2^{(1-\alpha)\cdot{|\vecu|}} \cdot (|\vecu| + |\Ff|)^{\Oh(1)}$ time by applying the standard dynamic programming algorithm for Set Cover (e.g.\ Theorem 6.1 in \cite{fpt-book}), that runs in $2^{|\vecu'|} \cdot (|\vecu'| + |\Ff'|)^{\Oh(1)}$ time.
    
    It follows that for sufficiently large $c$ and sufficiently small $\eps>0$, one can solve the Set Cover instance $(\vecu,\Ff,t)$ in $\Oh(2^{(1-\eps)\cdot{|\vecu|}} \cdot (|\vecu| + |\Ff|)^c)$ time.
\end{proof}

We now show an equivalence between Set Cover and $k$-Orthogonal Vectors in the setting of exact algorithms.

\begin{theorem}\label{thm:set-cover-equiv}
    The following statements are equivalent:
    \begin{itemize}
        \item[(a)] There exist $c > 0$ and $\eps > 0$, such that $k$-Orthogonal Vectors on $n$ vectors in dimension $d$ can be solved in $\Oh(2^{(1-\eps)\cdot d} \cdot n^c)$ time (algorithm works uniformly for a given $k \in \nat$).
        \item[(b)] There exists $c > 0$ and $\eps > 0$, such that Set Cover on a universe $\vecu$ and a set family $\Ff$ can be solved in $\Oh(2^{(1-\eps)\cdot{|\vecu|}} \cdot (|\vecu| + |\Ff|)^c)$ time.
    \end{itemize}
\end{theorem}
\begin{proof}
    \textbf{Implication (a)} $\bm{\Rightarrow}$ \textbf{(b)}.
    Let $(\vecu,\Ff,t)$ be an instance of Set Cover. Let $c>0$ and $\eps>0$ be fixed constants such that $t$-Orthogonal Vectors is solvable in $\Oh(2^{(1-\eps)\cdot d} \cdot n^c)$ time.
Consider the instance $\Aa_1, \ldots, \Aa_t$ of $t$-Orthogonal Vectors where $\Aa_1 = \ldots = \Aa_t$ and 
$\veca\subseteq \vecu$ is in $\Aa_1$ iff $\vecu\setminus \veca\in \Ff$. Since we constructed $|F|$ vectors of dimension $|\vecu|$, it suffices to show that there exists a solution for the Set Cover instance iff there exists a solution for the $t$-Orthogonal Vectors instance.

% It suffices to show that there exists a solution for the Set Cover instance iff there exists a solution for the $t$-Orthogonal Vectors instance.
If there exists a solution $\veca_1, \ldots, \veca_t$ for the $t$-Orthogonal
Vectors instance, then $\bigcap_{i=1}^{t} \veca_i = \emptyset$, meaning that
$\bigcup_{i=1}^{t} \vecu\setminus \veca_i = \vecu$. But $\vecu\setminus \veca_i \in \Ff$ by construction, therefore $\vecu \setminus \veca_1, \dots, \vecu \setminus \veca_t$ is a solution for the Set Cover instance. 
Similarly, if there exists a solution $\vecs_1, \ldots, \vecs_t$ for the Set Cover
instance, then $\bigcup_{i=1}^{t} \vecs_i = \vecu$, meaning that $\bigcap_{i=1}^{t}
\vecu\setminus \vecs_i = \emptyset$. But $\vecu\setminus \vecs_i \in \Aa_i$ by construction, therefore $\vecu \setminus \vecs_1, \dots, \vecu \setminus \vecs_t$ is a solution for the $t$-Orthogonal Vectors problem.

    \textbf{Implication (a)} $\bm{\Leftarrow}$ \textbf{(b)}.
Let $(A_1,\ldots,A_k)$ be an instance of $k$-Orthogonal Vectors of dimension $d$ and let $c >0$ and $\eps>0$ be fixed constants such that Set Cover is solvable in $\Oh(2^{(1-\eps)\cdot{|\vecu|}} \cdot (|\vecu| + |\Ff|)^{c})$ time. Consider the instance $(\vecu,\Ff,t)$ of Set Cover where $t \coloneq k$, $\vecu \coloneq [d+k]$, and $\Ff \coloneq \bigcup_{i \in [k]} \Ff_i$, where $\Ff_i \coloneq \{\{d+i\} \cup ([d] \setminus \veca) \ \mid \ \veca \in A_i\}$.
Since $k$ is a constant, the size of the universe is $|\vecu| = d + \Oh(1)$ and $F = \Oh(n)$. Hence, the running time of the Set Cover algorithm is $\Oh(2^{(1-\eps)\cdot{d}} \cdot n^c)$. It suffices to show that there exists a solution for the $k$-Orthogonal Vectors instance iff there exists a solution for the constructed Set Cover instance.

Let $\veca_1, \ldots, \veca_k$ be a solution for the $k$-Orthogonal Vectors
instance. Then $\bigcap_{i=1}^{k} \veca_i = \emptyset$, meaning that
$\bigcup_{i=1}^{k} [d] \setminus \veca_i = [d]$.
Therefore, the sets $\{\{d+i\} \cup ([d] \setminus \veca_i) \ \mid \
i \in [k]\}$ cover the universe $\vecu$, and thus there exists a solution for
the Set Cover instance.  Similarly, let $\vecs_1, \ldots, \vecs_k$ be a solution
for the Set Cover instance. As $t = k$ and for every $i \in [k]$ it holds that
$\{d+i\} \in \vecs$ iff $\vecs \in \Ff_i$, we can assume that $\vecs_i \in
\Ff_i$ for every $i \in [k]$.  Because $\vecs_1,\ldots\vecs_k$ is a valid cover,
we have $\bigcup_{i=1}^{k} \vecs_i = \vecu$, and thus $\bigcup_{i=1}^{k} (\vecs_i\setminus \{d+i\}) = [d]$.
Therefore $\bigcap_{i=1}^{k} [d] \setminus \vecs_i = \emptyset$, meaning that the sets $\{\vecu\setminus
\vecs_i \ \mid \ i \in [k]\}$ are a solution
for the $k$-Orthogonal Vectors instance.
\end{proof}

The faster algorithm for $k$-Orthogonal Vectors follows directly from \cref{cor:arc-set-cover} and \cref{thm:set-cover-equiv}.

\arcalg*

On the other hand, the {Set Cover Conjecture} implies that there is no $\eps>0$ for which Set Cover
is solvable in $\Os(2^{(1-\eps)\cdot{|\vecu|}})$ time. In particular, \cref{thm:set-cover-equiv} implies the lower bound for $k$-Orthogonal Vectors stated in \cref{thm:lowerBound}, based on the Set Cover Conjecture (which is incompatible with the Asymptotic Rank Conjecture~\cite{bjorklund24}).

\scclowerbound*

\bibliographystyle{abbrv}
\bibliography{bib}

\end{document}